\documentclass{article}

\usepackage{amsmath,amssymb,amsthm,amsfonts,graphicx,xcolor}
\usepackage{cancel}
\usepackage[margin=1in]{geometry}

\newtheorem{remark}{Remark}
\theoremstyle{definition}
\newtheorem{definition}{Definition}[section] 

\newtheorem{prop}{Proposition}
\newtheorem{corollary}{Corollary}

\newcommand{\real}{\mathcal{R}}

\newcommand{\tr}[1]{\mathbf{tr}\left(#1\right)}
\newcommand{\Exp}[1]{\mathbb{E}\left[#1\right]}
\newcommand{\Cov}[1]{\Exp{#1{#1}^T}}

\newcommand{\eqn}[1]{(\ref{eqn:#1})}
\newcommand{\eqnlabel}[1]{\label{eqn:#1}}

\title{Data Assimilation in Optimal Transport Framework (DRAFT)}
\author{Raktim Bhattacharya \\ Aerospace Engineering, Texas A\&M University.}
\begin{document}
\maketitle
%\section{Introduction}
%
%\comment{There is connection to multi-object tracking -- space debris, for example.}
%\comment{We can show linear Gaussian case boils down to Kalman filtering.}
%
%Layout of the paper is the following:
%\begin{itemize}
%\item We first develop the general framework for OT filtering for linear system
%\item Then show it is Kalman filtering to Gaussian system
%\item For non Gaussian, we show it is better than best "linear estimate". We can use Polynomial Chaos or something else to derive the filter. We can look at general exponential family in arbitrary manifolds $\real^m\times\sph^n$, and also consider mixture models. 
%\item Give an information geometric spin on this. IG tells us exponential family and mixture models are dual to each other.
%\end{itemize}

\section{Mathematical Preliminaries}
%\comment{Here we present the needed mathematical tools. }

\begin{definition} (\textbf{Wasserstein distance})
Consider the vectors $x_{1}, x_{2} \in \mathbb{R}^{n}$. Let $\mathcal{P}_{2}(p_1,p_2)$ denote the collection of all probability measures $p$ supported on the product space $\mathbb{R}^{2n}$, having finite second moment, with first marginal $p_{1}$ and second marginal $p_{2}$. The Wasserstein distance of order 2, denoted as $W_2$, between two probability measures $p_{1},p_{2}$, is defined as
\label{Wassdefn}
\begin{align}
&W_2(p_{1},p_{2}) \triangleq \left(\displaystyle\inf_{p\in\mathcal{P}_{2}(p_{1},p_{2})}\displaystyle\int_{\mathbb{R}^{2n}} \parallel x_{1}-x_{2}\parallel_{\ell_{2}\left(\mathbb{R}^{n}\right)}^{2} \: dp(x_{1},x_{2}) \right)
^{\frac{{1}}{2}}. \eqnlabel{W-dist}
\end{align}
\end{definition}
\begin{remark}
Intuitively, Wasserstein distance equals the least amount of work needed to morph one distributional shape to the other, and can be interpreted as the cost for Monge-Kantorovich optimal transportation plan \cite{villani2003topics}. The particular choice of $\ell_{2}$ norm with order 2 is motivated by \cite{halder2012further}. Further, one can prove (p. 208, \cite{villani2003topics}) that $W_2$ defines a metric on the manifold of PDFs.
\label{WassRemarkFirst}
\end{remark}
 
\begin{prop}
The Wasserstein distance $W_2$ between two multivariate Gaussians $\mathcal{N}(\mu_1, \Sigma_1)$ and $\mathcal{N}(\mu_2, \Sigma_2)$ in $\real^n$ is given by
\begin{align}   
W_2^2\left(\mathcal{N}(\mu_1, \Sigma_1),\mathcal{N}(\mu_2, \Sigma_2)\right)=\|\mu_1 - \mu_2\|^2+\tr{\Sigma_1+\Sigma_2-2\left(\sqrt{\Sigma_1}\Sigma_2\sqrt{\Sigma_1}\right)^{\frac{1}{2}}
}. \eqnlabel{wassdef}
\end{align}
\end{prop}

\begin{corollary} The Wasserstein distance between Gaussian distribution $\mathcal{N}(\mu, \Sigma)$  and the Dirac delta function $\delta(x-\mu_c)$ is given by,
\begin{align}\eqnlabel{W_dirac}
W_2^2(\mathcal{N}(\mu, \Sigma),\delta(x-\mu_c))=\|\mu-\mu_c\|^2+\tr{\Sigma}.
\end{align}
\end{corollary}

\begin{proof}    
Defining the Dirac delta function as (see e.g., p. 160-161, \cite{hassani1999mathematical})
\begin{align*}
   \delta(x-\mu_c)=\lim_{\mu\to \mu_c,\Sigma\to 0} \mathcal{N}(\mu,\Sigma),
\end{align*}
and substituting in \eqn{wassdef}, we get the result.
\end{proof}

The Wasserstein distance in \eqn{W_dirac} can also be written as,
\begin{align}
   W_2^2(\mathcal{N}(\mu, \Sigma),\delta(x-\mu_c))&=\|\mu-\mu_c\|^2+\tr{\Sigma},\nonumber \\
&={(\mu-\mu_c)^T(\mu-\mu_c)+\tr{\Sigma}
},\nonumber \\
&={\tr{(\mu-\mu_c)(\mu-\mu_c)^T}+\tr{\Sigma}
}, \nonumber  \\
&={\tr{(\mu-\mu_c)(\mu-\mu_c)^T+\Sigma}} \eqnlabel{main2}.
\end{align}

%\begin{prop}
%\comment{Show $W_2$ distance of MoG from dirac.}.
%\end{prop}
%
%\begin{prop}
%\comment{Show the LP problem that has to be solved to compute $W_2$ for a general PDF, represented by weighted particles, to the dirac}.
%\end{prop}

\section{Data Assimilation by Minimizing Wasserstein Distance}
Here we consider the problem of updating the prior state estimate with available measurement data to arrive at the posterior state estimate. We assume $x\in\real^n$ is the true state, a deterministic quantity. The prior estimate of $x$ is denoted by $x^-$, which is a random variable with associate probability density function $p_{x^-}(x^-)$. Similarly, the posterior estimate  is denoted by $x^+$, which is also a random variable with associated probability density function  $p_{x^+}(x^+)$. We next assume that measurement $y\in\real^m$ is a function of the true state $x$, but corrupted by an additive noise $n$. It is modeled as 
\begin{align}
y := g(x) + n. 
\end{align} 
The noise is assumed to be a random variable with associated probability density function $p_n(n)$.

The errors associated with the prior and posterior estimates are defined as 
\begin{align}
e^- &:= x^- - x,\\
e^+ &:= x^+ - x,
\end{align}
which are also random variables with probability density functions $p_{e^-}(e^-)$ and $p_{e^+}(e^+)$ respectively.

The objective of data assimilation is to determine the mapping $(x^-,y) \mapsto x^+$ such that some performance metric on $e^+$ is optimized. Abstractly, this can be written as

\begin{align}
\min_{T(x^-,y)} d(e^+), \eqnlabel{aopt}
\end{align}
where $T: (x^-,y) \mapsto x^+$, and $d(e^+)$ is some cost function to be minimized. 

In this paper, we define $d(e^+):= W^2_2(p_{e^+}(e^+),\delta(e^+))$, i.e. minimize the Wasserstein distance of the posterior-error's probability density function from the Dirac at the origin $\delta(e^+)$. The Dirac is a degenerate probability density function that represents \textit{zero error} in the posterior estimate. The optimization therefore determines the map $T(\cdot,\cdot)$ that takes $p_{e^+}(e^+)$ closest to $\delta(e^+)$ with respect to the Wasserstein distance, resulting in the best posterior estimate of $x$ in this sense.

In the next subsections, we present few specific cases of \eqn{aopt}, which are of engineering importance.

\subsection{Linear Measurement with Gaussian Uncertainty} 
Here we consider the simplest case, where the measurement model is linear and uncertainty is Gaussian. We show that we recover the classical Kalman update.
 
In this case, the joint probability density function associated with $e^+$ be Gaussian, denoted by $\mathcal{N}(\mu_e^+, \Sigma_e^+)$, where 
\begin{align}
\mu_e^+ &:= \Exp{e^+},
\end{align}
and
\begin{align}
\Sigma_e^+ &:= \Exp{(e^+ - \mu_e^+)(e^+ - \mu_e^+)^T} = \Exp{e^+e^{+^T}} - \mu_e^+\mu_e^{+^T}.
\end{align}
We also note that \textit{zero error} is associated with the degenerate probability density function, the Dirac delta function $\delta(e^+)$. The Wasserstein distance between $\mathcal{N}(\mu_e^+, \Sigma_e^+)$ and the Dirac delta function $\delta(e^+)$ is given by,
\begin{align}
W_2^2(\mathcal{N}(\mu_e^+, \Sigma_e^+),\delta(e^+))&=\tr{\mu_e^+{\mu_e^+}^T+\Sigma_e^+} = \tr{\Exp{e^+e^{+^T}}}.
\end{align}

We first define the sensor model as 
\begin{align}
y := Cx + n, \eqnlabel{meas}
\end{align}
where $n$ is a Gaussian random variable defined by $n \sim \mathcal{N}(0,R)$. Consequently, $y$ is also a Gaussian random variable.

We next define the random variable $x^+$ to be a linear combination of random variables $x^-$ and $y$, i.e.
\begin{align}
x^+ := Gx^- + Hy, \eqnlabel{posterior}
\end{align}
 and determine $G$ and $H$ such that $W_2^2(\mathcal{N}(\mu_e^+, \Sigma_e^+),\delta(e^+))$ is minimized, i.e.

$$
\min_{G,H} W_2^2(\mathcal{N}(\mu_e^+, \Sigma_e^+),\delta(e^+)),
$$
where all the posterior quantities are functions of $G$ and $H$. We next express $e^+$ as
\begin{align*}
e^+ &:= x^+-x , \\
&= Gx^- + Hy - x, \\
&= Gx^- + H(Cx+n) - x, \\
&= Ge^- + (G+HC-I)x + Hn.
\end{align*}

Noting that $\Exp{e^-} = 0$ (for unbiased prior estimate), $\Exp{e^-n^T} = 0$ (because $e^-$ and $n$ are uncorrelated), and $R:=\Cov{n}$, we can express $\Cov{e^+}$ as 

\begin{align*}
\Cov{e^+} &= \Cov{\Big(Ge^- + (G+HC-I)x + Hn\Big)},\\
&= G\Cov{e^-}G^T + (G+HC-I)xx^T(G+HC-I)^T + HRH^T,\\
& = G\Sigma_e^-G^T + (G+HC-I)xx^T(G+HC-I)^T + HRH^T. 
\end{align*}

Choosing $G:= I-HC$, reduces $\Cov{e^+}$ to
\begin{align*}
\Cov{e^+} &= (I-HC)\Sigma_e^-(I-HC)^T +  HRH^T, \\
&= \Sigma_e^- + H(C\Sigma_e^-C^T+R)H^T -HC\Sigma_e^- - \Sigma_e^-C^T H^T 
\end{align*}
which is further minimized by solving for $H$ satisfying 
$$
\frac{\partial\, \tr{\Cov{e^+}}}{\partial H} = 0,
$$
or 
$$
 H^\ast(C\Sigma_e^-C^T+R) - \Sigma_e^-C^T = 0,
$$
resulting in the optimal solution
\begin{align}
H^\ast := \Sigma_e^-C^T(C\Sigma_e^-C^T+R)^{-1},
\end{align}
which is also the optimal Kalman gain. Further, choosing $G:= I-HC$, also ensures $\Exp{e^+} = 0$, i.e. posterior error is unbiased. Therefore, we see that for linear measurement and Gaussian uncertainty models, minimization of the Wasserstein distance results in the classical Kalman update.

\bibliographystyle{unsrt}
\bibliography{Analysis_ref}
\end{document}